\documentclass[a4paper,UKenglish,cleveref, autoref,mathscr]{lipics-v2019}

\usepackage{bbm, tikz, mathtools, thm-restate}
\usetikzlibrary{arrows,calc,automata,intersections}
\tikzset{LMC style/.style={>=angle 60,every edge/.append style={thick},every state/.style={thick,minimum size=20,inner sep=0.5}}}

\usepackage{asymptote}

\newcommand{\RR}{\mathbb{R}}

\newcommand{\NN}{\mathbb{N}}

\newcommand{\QQ}{\mathbb{Q}}
\newcommand{\PP}{\mathbb{P}}

\newcommand{\GG}{\mathscr{G}}

\newcommand{\Exp}{\mathit{Exp}}

\newcommand{\1}{\mathbbm{1}}

\newcommand{\supp}{\mathrm{supp}}
\newcommand{\pl}{\Gamma_{\mathit{GEM}}}
\newcommand{\Leb}{\lambda_{\mathit{Leb}}}

\nolinenumbers

\graphicspath{ {images/} }

\DeclareMathOperator{\Span}{span\,}

\title{Equivalence of Hidden Markov Models with Continuous Observations}

\author{Oscar Darwin}{Department of Computer Science, Oxford University, United Kingdom }{}{https://orcid.org/0000-0001-5016-014X}{}

\author{Stefan Kiefer}{Department of Computer Science, Oxford University, United Kingdom}{}{https://orcid.org/0000-0003-4173-6877}{}

\authorrunning{O. Darwin and S. Kiefer}

\Copyright{John Q. Public and Joan R. Public}



\category{}

\relatedversion{}

\supplement{}

\EventEditors{John Q. Open and Joan R. Access}
\EventNoEds{2}
\EventLongTitle{42nd Conference on Very Important Topics (CVIT 2016)}
\EventShortTitle{CVIT 2016}
\EventAcronym{CVIT}
\EventYear{2016}
\EventDate{December 24--27, 2016}
\EventLocation{Little Whinging, United Kingdom}
\EventLogo{}
\SeriesVolume{42}
\ArticleNo{23}
\begin{document}

\maketitle

\begin{abstract}
We consider Hidden Markov Models that emit sequences of observations that are drawn from continuous distributions. For example, such a model may emit a sequence of numbers, each of which is drawn from a uniform distribution, but the support of the uniform distribution depends on the state of the Hidden Markov Model. Such models generalise the more common version where each observation is drawn from a finite alphabet. We prove that one can determine in polynomial time whether two Hidden Markov Models with continuous observations are equivalent.
\end{abstract}

\section{Introduction}

A (discrete-time, finite-state) \emph{Hidden Markov Model (HMM)} (often called \emph{labelled Markov chain}) has a finite set $Q$ of states and for each state a probability distribution over its possible successor states.
For any two states $q, q'$, whenever the state changes from $q$ to~$q'$, the HMM samples and then emits a random observation according to a probability distribution $D(q, q')$.
For example, consider the following diagram visualising a HMM:
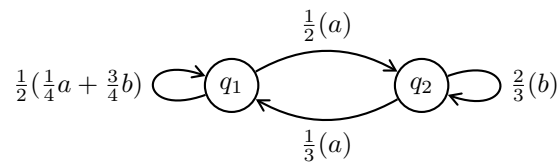
\begin{center}
\begin{tikzpicture}[scale=2.5,LMC style]
\node[state] (q1) at (0,0) {$q_1$};
\node[state] (q2) at (1,0) {$q_2$};
\path[->] (q1) edge [loop,out=200,in=160,looseness=10] node[left] {$\frac12 (\frac14 a + \frac34 b)$} (q1);
\path[->] (q1) edge [bend left] node[above] {$\frac12 (a)$} (q2);
\path[->] (q2) edge [loop,out=20,in=-20,looseness=10] node[right] {$\frac23 (b)$} (q2);
\path[->] (q2) edge [bend left] node[below] {$\frac13 (a)$} (q1);
\end{tikzpicture}
\end{center}
In state~$q_1$, the successor state is $q_1$ or~$q_2$, with probability~$\frac12$ each.
Upon transitioning from $q_1$ to itself, observation~$a$ is drawn with probability~$\frac14$ and observation~$b$ is drawn with probability~$\frac34$; upon transitioning from $q_1$ to~$q_2$, observation~$a$ is drawn surely.%
\footnote{One may allow for observations also on the states and not only on the transitions. But such state observations can be equivalently emitted upon leaving the state. Hence we can assume without loss of generality that all observations are emitted on the transitions.}

In this way, a HMM, together with an initial distribution on states, generates a random infinite sequence of observations.
In the example above, if the initial distribution is the Dirac distribution on~$q_1$, the probability that the observation sequence starts with~$a$ is $\frac12 \cdot \frac14 + \frac12$ and the probability that the sequence starts with~$a b$ is $\frac12 \cdot \frac14 \cdot \frac12 \cdot \frac34 + \frac12 \cdot \frac23$.

In the example above the observations are drawn from a finite observation alphabet $\Sigma = \{a,b\}$.
Indeed, in the literature HMMs most commonly have a finite observation alphabet.
In this paper we lift this restriction and consider \emph{continuous-observation} HMMs, by which we mean HMMs as described above, but with continuous observation set~$\Sigma$.
For example, instead of the distributions on~$\{a,b\}$ in the picture above (written there as $(\frac14 a + \frac34 b)$, $(a)$, $(b)$, respectively), we may have distributions on the real numbers. For example in the following diagram, where $U[a,b)$ denotes the uniform distribution on~$[a,b)$ and $\Exp(\lambda)$ denotes the exponential distribution with parameter~$\lambda$:
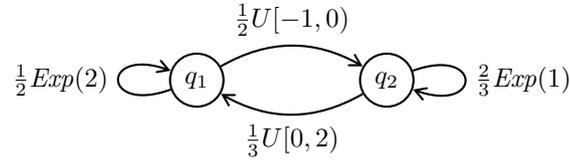
\begin{center}
\begin{tikzpicture}[scale=2.5,LMC style]
\node[state] (q1) at (0,0) {$q_1$};
\node[state] (q2) at (1,0) {$q_2$};
\path[->] (q1) edge [loop,out=200,in=160,looseness=10] node[left] {$\frac12 \Exp(2)$} (q1);
\path[->] (q1) edge [bend left] node[above] {$\frac12 U[-1,0)$} (q2);
\path[->] (q2) edge [loop,out=20,in=-20,looseness=10] node[right] {$\frac23 \Exp(1)$} (q2);
\path[->] (q2) edge [bend left] node[below] {$\frac13 U[0,2)$} (q1);
\end{tikzpicture}
\end{center}
HMMs, both with finite and infinite observation sets, are widely employed in fields such as speech recognition (see~\cite{Rabiner89} for a tutorial),
gesture recognition~\cite{Gesture},
signal processing~\cite{SignalProcessing},
and climate modeling~\cite{Weather}.
HMMs are heavily used in computational biology~\cite{HMM-comp-biology},
more specifically in DNA modeling~\cite{DNA-modeling} and biological sequence analysis~\cite{durbin1998biological},
including protein structure prediction~\cite{ProteinStructure} 
and gene finding~\cite{GeneFinding}.
In computer-aided verification, HMMs are the most fundamental model for probabilistic systems; model-checking tools such as Prism~\cite{KNP11} and Storm~\cite{Storm} are based on analyzing HMMs efficiently.

One of the most fundamental questions about HMMs is whether two HMMs with initial state distributions are \emph{(trace) equivalent}, i.e., generate the same distribution on infinite observation sequences.
For finite observation alphabets this problem is very well studied and can be solved in polynomial time using algorithms that are based on linear algebra~\cite{schut61,Paz71,Tzeng92,CortesMRdistance}.
Checking trace equivalence is used in the verification of obliviousness and anonymity, properties that are hard to formalize in temporal logics, see, e.g., \cite{EquivForSecurity10,kief11,ModVerif17}.

Although the generalisation to continuous observations (such as passed time, consumed energy, sensor readings) is natural, there has been little work on the algorithmics of such HMMs.
One exception is \emph{continuous-time} Markov chains (CTMCs) \cite{BHHK03,CDKM11} which are similar to HMMs described above, but with two kinds of observations: on the one hand they emit observations from a finite alphabet, but on the other hand they also emit the \emph{time} spent in each state. Typically, each state-to-state transition is labelled with a parameter~$\lambda$; for each transition its time of ``firing'' is drawn from an exponential distribution with parameter~$\lambda$; the transition with the smallest firing time ``wins'' and causes the corresponding change of state.
CTMCs have attractive properties: they are in a sense memoryless, and for many analyses, including model checking, an equivalent discrete-time model can be calculated using an efficient and numerically stable process called \emph{uniformization}~\cite{Grassmann91}.

In~\cite{HKK14} a stochastic model more general than ours was introduced, allowing not only for uncountable sets of observations (called \emph{labels} there), but also for infinite sets of states and actions.
The paper~\cite{HKK14} focuses on bisimulation; trace equivalence is not considered.
It emphasizes nondeterminism, a feature we do not consider here.

To the best of the authors' knowledge, this paper is the first to study equivalence of HMMs with continuous observations.
As continuous functions are part of the input, an equivalence checking algorithm, if it exists (which is not a priori clear), needs to be \emph{symbolic}, i.e., needs to perform computations on functions.
Our contributions are as follows:
\begin{enumerate}
\item
We show in \cref{sec-equivalence-as-orthogonality} that certain aspects of the linear-algebra based approach for checking equivalence of finite-observation HMMs carry over to the continuous case naturally.
In particular, equivalence reduces to orthogonality in a certain vector space of state-indexed real vectors, see \cref{equivifperpspan}.
\item
However, we show in \cref{finitelabred} that in the continuous case there can be \emph{additional} linear dependencies between the observation density functions (which is impossible in the finite case, where the different observations can be assumed linearly independent).
This renders a simple-minded reduction to the finite case incorrect.
Therefore, an equivalence checking algorithm needs to consider the interplay with the vector space from item~1.
\item
For the required computations on the observation density functions we introduce in \cref{sec-linearly-decomposable} \emph{linearly decomposable profile languages}, which are languages (i.e., sets of finite words) whose elements encode density functions on which basis computations can be performed efficiently.
In \cref{sub-profile-example} we provide an extensive example of such a language, encoding (linear combinations of) Gaussian, exponential, and piecewise polynomial density functions.
The proof that this language has the required properties is non-trivial itself and requires \emph{alternant matrices} and comparisons of the tails of various density functions.
\item
In \cref{starredsec} we finally show that HMMs whose observation densities are given in terms of linearly decomposable profile languages can be checked for equivalence in \emph{polynomial time}, by a reduction to the finite-observation case.
We also indicate, in \cref{ex-timing}, how our result can be used to check for susceptibility of certain timing attacks.
\end{enumerate}

%


\section{Preliminaries}
We write $\NN$ for the set of positive integers, $\QQ$ for the set of rationals and $\QQ_+$ for the set of positive rationals.
For $d \in \NN$ and a finite set $Q$ we use the notation $|Q|$ for the number of elements in $Q$, $[d] = \{1, \dots, d\}$ and $[Q] = \{1, \dots, |Q|\}$. Vectors $\mu \in \RR^N$ are viewed as row vectors and we write $\1 = (1, \dots, 1) \in \RR^N$.
Superscript~$T$ denotes transpose; e.g., $\1^T$ is a column vector of ones.
A matrix $M \in \RR^{N \times N}$ is \emph{stochastic} if $M$ is non-negative and $\sum_{j = 1}^{N} M_{i,j} = 1$ for all $i \in [N]$.
For a domain $\Sigma$ and subset $E \subseteq \Sigma$ the \emph{characteristic} function $\chi_E : \Sigma \rightarrow \{0,1\}$ is defined as $\chi_E(x) = 1$ if $x \in E$ and $\chi_E(x) = 0$ otherwise.

Throughout this paper, we use $\Sigma$ to denote a set of \emph{observations}.
We assume $\Sigma$ is a topological space and $(\Sigma, \GG, \lambda)$ is a measure space where all the open subsets of $\Sigma$ are contained within $\GG$ and have non-zero measure. Indeed $\RR$ and the usual Lebesgue measure space on $\RR$ satisfy these assumptions.
The set $\Sigma^n$ is the set of words over~$\Sigma$ of length $n$ and $\Sigma^* = \bigcup_{n = 0}^\infty \Sigma^n$.

A matrix valued function $\Psi : \Sigma \rightarrow [0,\infty)^{N \times N}$ can be integrated element-wise.
We write $\int_E \Psi\,d\lambda$ for the matrix with entries $\left( \int_E \Psi\, d\lambda \right)_{i,j} = \int_E \Psi_{i,j}\, d\lambda$, where $\Psi_{i,j} : \Sigma \rightarrow [0,\infty)$ is defined by $\Psi_{i,j}(x) = \big( \Psi(x) \big)_{i,j}$ for all $x \in \Sigma$.

A function $f : \Sigma \rightarrow \RR^m$ is \emph{piecewise continuous} if there is an open set $C \subseteq \Sigma$, called a \emph{set of continuity}, such that $f$ is continuous on $C$ and for every point $x \in  \Sigma \setminus C$ there is some sequence of points $x_n \in C$ such that $\lim_{n \rightarrow \infty} x_n = x$ and $\lim_{n \rightarrow \infty} f(x_n) = f(x)$. For a non-negative function $f : \Sigma \rightarrow [0,\infty)$ we use the notation ${\supp~f = \{x \in \Sigma \mid f(x) > 0\}}$.

\begin{definition}\label{HMMdef}
A \emph{Hidden Markov Model} (HMM) is a triple $(Q, \Sigma, \Psi)$ where $Q$ is a finite set of states, $\Sigma$ is a set of observations, and the \emph{observation density matrix} $\Psi : \Sigma \rightarrow [0,\infty)^{|Q| \times |Q|}$ specifies the transitions such that $\int_\Sigma \Psi\, d\lambda$ is a stochastic matrix.
\end{definition}
\begin{example} \label{ex-HMMdef}
The second HMM from the introduction is the triple $(\{q_1, q_2\}, \mathbb{R}, \Psi)$ with
\begin{equation}
\Psi(x) \ = \ \begin{pmatrix}
\frac12 \cdot 2 \exp(-2 x) \cdot \chi_{[0,\infty)}(x) && \frac12 \cdot 1 \cdot \chi_{[-1,0)}(x) \\
\frac13 \cdot \frac12 \cdot \chi_{[0,2)}(x) && \frac23 \cdot \exp(-x) \cdot \chi_{[0,\infty)}(x)
\end{pmatrix}\,. \tag*{\qed}
\end{equation}
\end{example}
We assume that $\Psi$ is piecewise continuous and extend $\Psi$ to the mapping $\Psi : \Sigma^* \rightarrow [0,\infty)^{|Q| \times |Q|}$ with $\Psi(x_1 \cdots x_n) = \Psi(x_1) \times \dots \times \Psi(x_n)$ for $x_1, \dots, x_n \in \Sigma$. If $C$ is the set of continuity for $\Psi : \Sigma \rightarrow [0,\infty)^{|Q| \times |Q|}$, then for fixed $n \in \NN$ the restriction $\Psi : \Sigma^n \rightarrow [0,\infty)^{|Q| \times |Q|}$ is piecewise continuous with set of continuity $C^n$. We say that $A \subseteq \Sigma^n$ is a \emph{cylinder set} if $A = A_1 \times \dots \times A_n$ and $A_i \in \GG$ for $i \in [n]$. For every $n$ there is an induced measure space $(\Sigma^n, \GG^n, \lambda^n)$ where $\GG^n$ is the smallest $\sigma$-algebra containing all cylinder sets in~$\Sigma^n$ and $\lambda^n(A_1 \times \dots \times A_n) = \prod_{i = 1}^n \lambda(A_i)$ for any cylinder set $A_1 \times \dots \times A_n$. Let $A \subseteq \Sigma^n$ and write $A \Sigma^\omega$ for the set of infinite words over~$\Sigma$ where the first $n$ observations fall in the set~$A$. Given a HMM $(Q, \Sigma, \Psi)$ and initial distribution $\pi$ on $Q$ viewed as vector $\pi \in \RR^{|Q|}$, there is an induced probability space $(\Sigma^\omega, \GG^*, \PP_\pi)$ where $\Sigma^\omega$ is the set of infinite words over~$\Sigma$, and $\GG^*$ is the smallest $\sigma$-algebra containing (for all $n \in \NN$) all sets $A \Sigma^\omega$ where $A\subseteq \Sigma^n$ is a cylinder set and $\PP_\pi$ is the unique probability measure such that
$\PP_\pi(A \Sigma^\omega) =  \pi \int_A \Psi\, d\lambda^n \1^T$
for any cylinder set $A \subseteq \Sigma^n$.
\begin{definition}
For two distributions $\pi_1$ and $\pi_2$ and a HMM $C = (Q, \Sigma, \Psi)$, we say that $\pi_1$ and~$\pi_2$ are \emph{equivalent}, written $\pi_1 \equiv_C \pi_2$, if $\PP_{\pi_1}(A) = \PP_{\pi_2}(A)$ holds for all measurable subsets $A \subseteq \Sigma^\omega$.
\end{definition}
One could define equivalence of two pairs $(C_1,\pi_1)$ and $(C_2,\pi_2)$ where $C_i = (Q_i, \Sigma, \Psi_i)$ are HMMs and $\pi_i$ are initial distributions for $i=1,2$.
We do not need that though, as we can define, in a natural way, a single HMM over the disjoint union of $Q_1$ and~$Q_2$ and consider instead equivalence of $\pi_1$ and~$\pi_2$ (where $\pi_1,\pi_2$ are appropriately padded with zeros).

Given an observation density matrix $\Psi$, a \emph{functional decomposition} consists of functions $f_k : \Sigma \rightarrow [0,\infty)$ and matrices $P_k \in \RR^{|Q| \times |Q|}$ for $k \in [d]$ such that $\Psi(x) = \sum_{k = 1}^d f_k(x) P_k$ for all $x \in \Sigma$ and $\int_{\Sigma} f_k\, d\lambda = 1$ for all $k \in [d]$. We sometimes abbreviate this decomposition as $\Psi = \sum_{k = 1}^d f_k P_k$ and this notion has a central role in our paper.

\begin{example} \label{ex-functional-decomposition}
The observation density matrix~$\Psi$ from \cref{ex-HMMdef} has a functional decomposition
\begin{align*}
\Psi(x) \ = \
& 2 \exp(-2 x) \chi_{[0,\infty)}(x)
\begin{pmatrix}
\frac12 && 0 \\ 0 && 0
\end{pmatrix}
+
\chi_{[-1,0)}(x)
\begin{pmatrix}
0 && \frac12 \\ 0 && 0
\end{pmatrix}
+ \mbox{} \\
& \frac12 \chi_{[0,2)}(x)
\begin{pmatrix}
0 && 0 \\
\frac13 && 0
\end{pmatrix}
+
\exp(-x) \chi_{[0,\infty)}(x)
\begin{pmatrix}
0 && 0 \\
0 && \frac23
\end{pmatrix}
\,. \tag*{\qed}
\end{align*}
\end{example}

\begin{lemma}\label{stochasticPk}
Let $(Q, \Sigma, \Psi)$ be a HMM.
If $\Psi$ has functional decomposition $\Psi =  \sum_{k = 1}^d f_k P_k$ then $\sum_{k = 1}^d P_k$ is stochastic.
\end{lemma}
\begin{proof}
By definition of a HMM, $\int_{\Sigma} \Psi\, d\lambda$ is stochastic, and we have
\begin{equation*}
\int_{\Sigma} \Psi\, d\lambda = \int_{\Sigma} \sum_{k = 1}^d f_k P_k\, d\lambda =  \sum_{k = 1}^d P_k  \int_{\Sigma} f_k\, d\lambda =  \sum_{k = 1}^d P_k.\qedhere
\end{equation*}
\end{proof}
When $\Sigma$ is finite, it follows that $\int_\Sigma \Psi\, d\lambda = \sum_{a \in \Sigma} \Psi(a)$. Hence $\sum_{a \in \Sigma} \Psi(a)$ is stochastic.

\subparagraph*{Encoding}
For computational purposes we assume that rational numbers are represented as ratios of integers in binary.
The initial distribution of a HMM with state set~$Q$ is given as a vector $\pi \in \QQ^{|Q|}$.
We also need to encode continuous functions, in particular, density functions such as Gaussian, exponential or piecewise-polynomial functions.
A \emph{profile} is a finite word (i.e., string) that describes a continuous function.
It may consist of (an encoding of) a function type and its parameters. For example, the profile  $(\mathcal{N}, \mu, \sigma)$ may denote a Gaussian (also called normal) distribution with mean $\mu \in \QQ$ and standard deviation $\sigma \in \QQ_+$. A profile may also consist of a description of a rational linear combination of such building blocks.
For any profile~$\gamma$ we write $[\![\gamma]\!] : \Sigma \rightarrow [0,\infty)$ for the function it encodes.
For example, a profile $\gamma = (\mathcal{N}, \mu, \sigma)$ with $\mu \in \QQ,\ \sigma \in \QQ_+$ may encode the function $[\![\gamma]\!] : \RR \rightarrow [0,\infty)$ given as $[\![\gamma]\!](x) = \frac{1}{\sigma\sqrt{2\pi}} \exp{- \frac{(x - \mu)^2}{2\sigma^2}}$.
Without restricting ourselves to any particular encoding, we assume that $\Gamma$ is a \emph{profile language}, i.e., a finitely presented but usually infinite set of valid profiles. For any $\Gamma_0 \subseteq \Gamma$ we write $[\![\Gamma_0]\!] = \{[\![\gamma]\!] \mid \gamma \in \Gamma_0\}$.

We use profiles to encode HMMs $C = (Q, \Sigma, \Psi)$:
we say that $C$ is \emph{over}~$\Gamma$ if the observation density matrix~$\Psi$ is given as a matrix of pairs $(p_{i,j}, \gamma_{i,j}) \in \QQ_+  \times \Gamma$ such that $\Psi_{i,j} = p_{i,j} [\![\gamma_{i,j}]\!]$ and $\int_{\Sigma} [\![\gamma_{i,j}]\!]\,d\lambda = 1$ hold for all $i,j \in [Q]$. In this way the $p_{i,j}$ form the transition probabilities between states and the $\gamma_{i,j}$ encode the probability densities of the observations upon each transition.

\begin{example} \label{ex-encoding-prelims}
For a suitable profile language~$\Gamma$, the HMM from \cref{ex-HMMdef} may be over~$\Gamma$, with the observation density matrix given as
\begin{equation}
\begin{pmatrix}
(\frac12, (\Exp,2)) && (\frac12, (U,-1,0)) \\
(\frac13, (U,0,2))  && (\frac23, (\Exp,1))
\end{pmatrix}\,. \tag*{\qed}
\end{equation}
\end{example}
The observation density matrix~$\Psi$ of a HMM $(Q, \Sigma, \Psi)$ with \emph{finite}~$\Sigma$ can be given as a
list of matrices $\Psi(a) \in \QQ_+^{|Q| \times |Q|}$ for all $a \in \Sigma$ such that $\sum_{a \in \Sigma} \Psi(a)$ is a stochastic matrix.

\section{Equivalence as Orthogonality} \label{sec-equivalence-as-orthogonality}

For finite-observation HMMs it is well known~\cite{schut61,Paz71,Tzeng92,CortesMRdistance} that two initial distributions given as vectors $\pi_1, \pi_2 \in \RR^{|Q|}$ are equivalent if and only if $\pi_1 - \pi_2$ is orthogonal (written as~$\perp$) to a certain vector space.
Indeed, this property holds more generally:

\begin{restatable}{proposition}{equivifperpspan}\label{equivifperpspan}
Consider a HMM $(Q, \Sigma, \Psi)$.
For any $\pi_1, \pi_2 \in \RR^{|Q|}$ we have
\[\pi_1 \equiv \pi_2 \ \iff \ \pi_1 - \pi_2 \perp \Span \{\Psi(w)\1^T \mid w \in \Sigma^*\}.\]
\end{restatable}
In the finite-observation case, \cref{equivifperpspan} leads to an efficient algorithm for deciding equivalence: it suffices to compute a basis for $\mathcal{V} = \Span \{\Psi(w)\1^T \mid w \in \Sigma^*\}$.
This can be done using a fixed-point algorithm that computes a sequence of (bases of) increasing subspaces of~$\mathcal{V}$: start with $\mathcal{B} = \{\1^T\}$, and as long as there is $a \in \Sigma$ and $v \in \mathcal{B}$ such that $\Psi(a) v \not\in \Span \mathcal{B}$, add $\Psi(a) v$ to~$\mathcal{B}$.
Since $\dim \mathcal{V} \le |Q|$, this algorithm terminates after at most $|Q|$ iterations, and returns $\mathcal{B}$ such that $\Span \mathcal{B} = \mathcal{V}$.
It is then easy to check whether $\pi_1 - \pi_2 \perp \mathcal{V}$.
It follows:
\begin{proposition} \label{prop-finite-HMM}
Given a HMM $(Q, \Sigma, \Psi)$ with finite~$\Sigma$ and initial distributions $\pi_1, \pi_2 \in \QQ^{|Q|}$, it is decidable in polynomial time whether $\pi_1 \equiv \pi_2$.
\end{proposition}
This is not an effective algorithm when $\Sigma$ is infinite.

\section{Labelling Reductions}\label{finitelabred}
Our goal is to reduce in polynomial time the equivalence problem in continuous-observation HMMs to the equivalence problem in finite-observation HMMs.
Since the latter is decidable in polynomial time by \cref{prop-finite-HMM}, a polynomial time algorithm for deciding equivalence in continuous-observation HMMs follows.

Towards this objective, consider a reduction where each continuous density function is given a label and these labels form the observation alphabet of a finite-observation HMM. For example consider the chain on the left in the diagram below. This disconnected HMM emits letters from two distinct normal distributions with profiles $(\mathcal{N}, 0, 1)$ and $(\mathcal{N}, 1, 2)$. Assigning each distribution letters $a,b$ respectively yields the HMM given on the right.
Since in the right chain states $q_1$ and $q_2$ are equivalent so too are the same labelled states in the continuous chain.

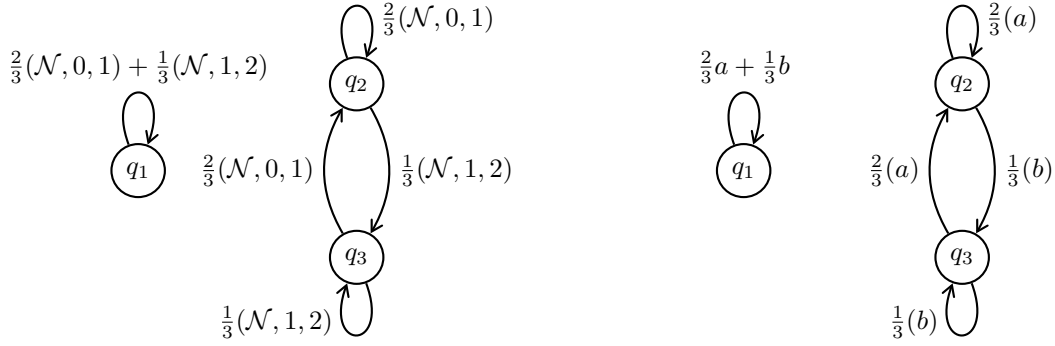
\begin{center}
	\begin{tikzpicture}[scale=2.3,LMC style]
	\node[state] (q1) at (-0.25,0) {$q_1$};
	\node[state] (q2) at (1,0.5) {$q_2$};
	\node[state] (q3) at (1,-0.5) {$q_3$};
	\path[->] (q1) edge [loop,out=110,in=70,looseness=10] node[above] {$\frac23 (\mathcal{N}, 0, 1) + \frac13 (\mathcal{N}, 1, 2)$} (q1);
	
	\path[->] (q3) edge [bend left] node[left] {$\frac23 (\mathcal{N}, 0, 1)$} (q2);
	\path[->] (q2) edge [bend left] node[right] {$\frac13 (\mathcal{N}, 1, 2)$} (q3);
	\path[->] (q3) edge [loop,out=290,in=250,looseness=10] node[pos=0.75,left] {$\frac13 (\mathcal{N}, 1, 2)$} (q3);
	\path[->] (q2) edge [loop,out=110,in=70,looseness=10] node[pos=0.75,right] {$\frac23 (\mathcal{N}, 0, 1)$} (q2);
	\end{tikzpicture}
\hfill
	\begin{tikzpicture}[scale=2.3,LMC style]
	\node[state] (q4) at (2.75,0) {$q_1$};
	\node[state] (q5) at (4,0.5) {$q_2$};
	\node[state] (q6) at (4,-0.5) {$q_3$};
	\path[->] (q4) edge [loop,out=110,in=70,looseness=10] node[above] {$\frac23 a + \frac13 b$} (q4);
	
	\path[->] (q6) edge [bend left] node[left] {$\frac23 (a)$} (q5);
	\path[->] (q5) edge [bend left] node[right] {$\frac13 (b)$} (q6);
	\path[->] (q6) edge [loop,out=290,in=250,looseness=10] node[pos=0.75,left] {$\frac13 (b)$} (q6);
	\path[->] (q5) edge [loop,out=110,in=70,looseness=10] node[pos=0.75,right] {$\frac23 (a)$} (q5);
	\end{tikzpicture}
\end{center}
More rigorously, if $C = (Q, \Sigma, \Psi)$ is a HMM over $\Gamma = \{\beta_1, \ldots, \beta_K\}$ and $\Psi$ is encoded as a matrix of coefficient-profile pairs $(p_{i,j}, \gamma_{i,j}) \in \QQ_+ \times \Gamma$ then we call the \emph{labelling reduction} the HMM $(Q, \hat{\Sigma}, \hat{M})$ where $\hat{\Sigma} = \{a_1, \dots, a_K\}$ is an alphabet of fresh observations and 
\[
\hat{M}_{i,j}(a_k) = \begin{cases}
p_{i,j} & \gamma_{i,j} = \beta_k \\
0 & \text{otherwise}.
\end{cases}
\]
Since $\Psi$ has functional decomposition $\Psi = \sum_{k = 1}^K [\![\beta_k]\!] \hat{M}(a_k)$, it follows by \Cref{stochasticPk} that $\sum_{k = 1}^K \hat{M}(a_k)$ is stochastic and the labelling reduction is a well defined HMM which may be computed in polynomial time. As discussed in the previous example, equivalence in the labelling reduction implies equivalence in the original chain:

\begin{proposition}\label{redbydist}
Let $C = (Q, \Sigma, \Psi)$ be a HMM with labelling reduction $L = (Q, \hat{\Sigma}, \hat{M})$. Then for any initial distributions $\pi_1$ and $\pi_2$
\[\pi_1 \equiv_L \pi_2 \implies \pi_1 \equiv_C \pi_2.\]
\end{proposition}
For the proof of \Cref{redbydist} we use the following lemma 
which will be re-used in \cref{starredsec}.

\begin{restatable}{lemma}{langequiv}\label{langequiv}
Let $C_1 = (Q, \Sigma_1, \Psi_1)$ and $C_2 = (Q, \Sigma_2, \Psi_2)$ be two HMMs with the same state space~$Q$.
Suppose that $\Span \{\Psi_1(x) \mid x \in \Sigma_1\} \subseteq \Span \{\Psi_2(x) \mid x \in \Sigma_2\}$.
Then, for any two initial distributions $\pi_1$ and $\pi_2$,
\[\pi_1 \equiv_{C_2} \pi_2 \implies \pi_1 \equiv_{C_1} \pi_2.\]
\end{restatable}
\begin{proof}[Proof of \cref{redbydist}]
$\Psi$ has a functional decomposition $\Psi = \sum_{k = 1}^K [\![\beta_k]\!] \hat{M}(a_k)$.
Thus, $\Span \{\Psi(x) \mid x \in \Sigma\} \subseteq \Span \{\hat{M}(a_k) \mid a_k \in \hat{\Sigma}\}$
and the statement follows by \cref{langequiv}.
\end{proof}

\begin{example}\label{finiteredproblem}
Consider the HMMs in the diagram below.
The HMM on the left is a continuous-observation chain where $D$ and $D'$ are distributions on $[0,1]$ with probability density functions $2x \chi_{[0,1)}(x)$ and $2(1 - x)\chi_{[0,1)}(x)$ respectively, and $U[a,b)$ is the uniform distribution on $[a,b)$. The HMM on the right is the corresponding labelling reduction.

Since $U[0,1) = \frac{1}{2} D + \frac{1}{2} D'$, (the Dirac distributions on) states $q_1$ and $q_4$ are equivalent but as the distributions $U[0,1), D, D'$ are distinct, they get assigned different labels $a, b, c$, respectively in the labelling reduction. The states $q_1$ and $q_4$ are therefore not equivalent in the right chain.
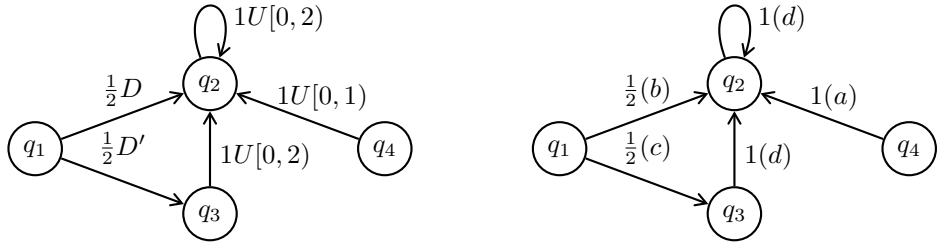
\begin{center}
	\begin{tikzpicture}[scale=2.3,LMC style]
	\node[state] (q1) at (0,0) {$q_1$};
	\node[state] (q2) at (1,0.375) {$q_2$};
	\node[state] (q3) at (1,-0.375) {$q_3$};
	\node[state] (q4) at (2,0) {$q_4$};
	\path[->] (q2) edge [loop,out=110,in=70,looseness=10] node[pos=0.75,right] {$1 U[0,2)$} (q2);
	\path[->] (q4) edge node[pos=0.75,right,yshift=1mm] {$1 U[0,1)$} (q2);
	\path[->] (q3) edge node[right,pos=0.4] {$1 U[0,2)$} (q2);
	\path[->] (q1) edge node[above] {$\frac{1}{2} D$} (q2);
	\path[->] (q1) edge node[above,yshift=1mm] {$\frac{1}{2} D'$} (q3);
	
	\node[state] (q5) at (3,0) {$q_1$};
	\node[state] (q6) at (4,0.375) {$q_2$};
	\node[state] (q7) at (4,-0.375) {$q_3$};
	\node[state] (q8) at (5,0) {$q_4$};
	\path[->] (q6) edge [loop,out=110,in=70,looseness=10] node[pos=0.75,right] {$1 (d)$} (q6);
	\path[->] (q8) edge node[above,pos=0.4] {$1 (a)$} (q6);
	\path[->] (q7) edge node[right,pos=0.4] {$1 (d)$} (q6);
	\path[->] (q5) edge node[above] {$\frac{1}{2} (b)$} (q6);
	\path[->] (q5) edge node[above,yshift=1mm] {$\frac{1}{2} (c)$} (q7);
	\end{tikzpicture}
\end{center}
\end{example}

\section{Linearly Decomposable Profile Languages} \label{sec-linearly-decomposable}

\Cref{finiteredproblem} shows that the linear combination of two continuous distributions can ``imitate'' a single distribution. Therefore we consider the transition densities as part of a vector space of functions. In the usual way $\mathcal{L}_1(\Sigma, \lambda)$ is the quotient vector space where functions that differ only on a $\lambda$-null set are identified. In particular, when $\Sigma \subseteq \RR$ and $\lambda$ is the Lebesgue measure~$\Leb$, the functions $\chi_{[a,b)}$ and $\chi_{(a,b]}$ are considered the same.

Let $\Gamma$ be a profile language with $[\![\Gamma]\!] \subseteq \mathcal{L}_1(\Sigma, \lambda)$. We say that $\Gamma$ is \emph{linearly decomposable} if for every finite set $\{\gamma_1, \dots, \gamma_n\} = \Gamma_0 \subseteq \Gamma$ one can compute in polynomial time profiles $\beta_1, \dots, \beta_m \in \Gamma_0$ such that $\{[\![\beta_1]\!], \dots, [\![\beta_m]\!]\}$ is a basis for $\Span \{[\![\gamma_1]\!], \dots, [\![\gamma_n]\!]\}$ (hence $m \leq n$), and further a set of coefficients $b_{i,j} \in \QQ$ for $i \in [n], j \in [m]$ such that
\begin{equation*}
[\![\gamma_i]\!] = \sum_{j = 1}^m b_{i,j}[\![\beta_j]\!] \text{ for all $i \in [n]$.}
\end{equation*}
The following theorem is the main result of this paper:
\begin{theorem}\label{computationequivalencethm}
Given a HMM $(Q, \Sigma, \Psi)$ over a linearly decomposable profile language, and initial distributions $\pi_1, \pi_2 \in \QQ^{|Q|}$, it is decidable in polynomial time (in the size of the encoding) whether $\pi_1 \equiv \pi_2$.
\end{theorem}
We prove \cref{computationequivalencethm} in \cref{starredsec}.
To make the notion of linearly decomposable profile languages more concrete, we give a concrete example in the following subsection.

\subsection{Example: Gaussian, Exponential, and Piecewise Polynomial Functions} \label{sub-profile-example}

We describe a profile language, $\pl$, that can specify linear combinations of Gaussian, exponential, and piecewise polynomial density functions.

We call a function of the form $x \mapsto x^k \chi_{I}(x)$ where $k \in \NN \cup \{0\}$ and $I \subset \RR$ is an interval an \emph{interval-domain monomial}.
To avoid clutter, we often denote interval-domain monomials only by $x^k \chi_I$.
Recall that $\mathcal{L}_1(\RR, \Leb)$ is a quotient space, so half open intervals~$I = [a,b)$ are sufficient. Any piecewise polynomial is a linear combination of interval-domain monomials.

Let $M$ be a set of profiles encoding interval-domain monomials $x^k \chi_{[a,b)}$ in terms of $k \in \NN \cup \{0\}$ and $a,b \in \QQ$. Gaussian and exponential density functions can be fully described using their parameters, which we assume to be rational. We write $G$ and $E$ for corresponding sets of profiles, respectively.
Finally, we fix a profile language $\pl \supset G \cup E \cup M$ obtained by closing $G \cup E \cup M$ under linear combinations. That is, for any $\gamma_1, \dots, \gamma_k \in \pl$ and $\lambda_1, \dots, \lambda_k \in \QQ$, there exists a profile $\gamma \in \pl$ such that $[\![\gamma]\!] = \lambda_1 [\![\gamma_1]\!] + \dots + \lambda_k [\![\gamma_k]\!]$.
This closure can be achieved using a specific constructor, say $\mathcal{S}$, for linear combinations, so that $\gamma = \mathcal{S}(\lambda_1, \gamma_1, \ldots, \lambda_k, \gamma_k)$.

\begin{example}\label{profileexample}
The HMM $(Q, \RR, \Psi)$ from \Cref{finiteredproblem} is over~$\pl$: the observation density matrix~$\Psi$ can be encoded as a matrix of coefficient-profile pairs
\[
\begin{pmatrix}
0 & (\frac12, \gamma_1)  & (\frac12,\gamma_2) & 0 \\
0 & (1,\gamma_3) & 0 & 0 \\
0 & (1,\gamma_3) & 0 & 0 \\
0 & (1,\gamma_4) & 0 & 0
\end{pmatrix}
\]
with $\gamma_1, \gamma_2, \gamma_3, \gamma_4 \in \pl$ and
$[\![\gamma_1]\!] = 2x\chi_{[0,1)}$ and
$[\![\gamma_2]\!] = 2(1-x)\chi_{[0,1)}$ and
$[\![\gamma_3]\!] = \frac12 \chi_{[0,2)}$ and
$[\![\gamma_4]\!] = \chi_{[0,1)}$.
\qed
\end{example}

\begin{restatable}{lemma}{linindepofcommonfuncs}\label{linindepofcommonfuncs}
Let $H$ be a set of disjoint half open intervals.
Suppose that $m_1, \dots, m_I$ are distinct interval-domain monomials such that $\supp~m_i \in H$ for all $i \in [I]$. In addition, let $g_1, \dots, g_J$ and $e_1, \dots, e_K$ be distinct Gaussian and exponential density functions, respectively. Then, the set $\{m_1, \dots, m_I, g_1, \dots, g_J, e_1, \dots, e_K\}$ is linearly independent.
\end{restatable}

For the proof of this lemma we need a result concerning \emph{alternant matrices}. Consider functions $f_1, \dots, f_n : \Sigma \rightarrow \RR$ and let $x_1, \dots, x_n \in \Sigma$. Then,
\[M = \begin{pmatrix}
f_1(x_1) &  f_2(x_1) &  \cdots  & f_n(x_1) \\
f_1(x_2) & f_2(x_2) & \cdots  & f_n(x_2) \\
\vdots & \vdots & \ddots & \vdots \\
f_1(x_n) & f_2(x_n) & \dots  & f_n(x_n)
\end{pmatrix}\]
is called the alternant matrix for $f_1, \dots, f_n$ and \emph{input points} $x_1, \dots, x_n$.

\begin{restatable}{lemma}{alternantexistence}\label{alternantexistence}
Suppose $f_1, \dots, f_n \in \mathcal{L}_1(\Sigma, \lambda)$.
Then, the $f_i$ are linearly dependent if and only if all alternant matrices for the $f_i$ are singular.
\end{restatable}

\begin{proof}[Sketch proof of \Cref{linindepofcommonfuncs}]
Under the assumption that a linear combination exists almost surely equal to $0$, by examining the limit at $+\infty$ we show that the exponential and Gaussian coefficients are zero. Then, by constructing an appropriate alternant matrix with full rank we invoke \Cref{alternantexistence} which means the remaining interval-domain monomials are linearly independent and thus must also have zero coefficients.
\end{proof}

\begin{restatable}{proposition}{commonfuncslineardecomp}\label{commonfuncslineardecomp}
The profile language $\pl$ is linearly decomposable.
\end{restatable} 

Thus we obtain the following corollary of \cref{computationequivalencethm}:
\begin{corollary}
Given a HMM $(Q, \Sigma, \Psi)$ over~$\pl$, and initial distributions $\pi_1, \pi_2 \in \QQ^{|Q|}$, it is decidable in polynomial time whether $\pi_1 \equiv \pi_2$.
\end{corollary}

\section{Proof of \Cref{computationequivalencethm}}\label{starredsec}

Suppose that $\Psi$ has a functional decomposition $\sum_{k = 1}^d f_k P_k$ such that the set $\{f_1, \dots, f_d\}$ is linearly independent. Then, $\sum_{k = 1}^d f_k P_k$ is called an \emph{independent functional decomposition}.
The efficient computation of an independent functional decomposition is the key ingredient for the proof of \cref{computationequivalencethm}. 
We start with the following lemma.

\begin{restatable}{lemma}{spanpreserving}\label{spanpreserving}
Suppose $\Psi : \Sigma \rightarrow [0,\infty)^{|Q| \times |Q|}$ has an independent functional decomposition $\Psi = \sum_{k = 1}^d f_k P_k$. Then, $\Span \{\Psi(x) \mid x \in \Sigma\} = \Span \{P_k \mid k \in [d]\}.$
\end{restatable}

\begin{proof}
Since $\Psi(x) = \sum_{k = 1}^d f_k(x) P_k$, we have $\Span \{\Psi(x) \mid x \in \Sigma\} \subseteq \Span \{P_k \mid k \in [d]\}$. For the reverse inclusion, since the $f_i$ are linearly independent, by \Cref{alternantexistence} there exists an alternant matrix $M$ with full rank for $f_1, \dots, f_d$ with input points $x_1, \dots, x_d$.
Hence, for each of the standard basis vectors $e_k \in \{0,1\}^d$, $k \in [d]$, there exists $v_k = (v_{1,k} ,\dots, v_{d,k}) \in \RR^d$ such that $v_k M = e_k$. Writing $\delta_{j,k}$ for the Kronecker delta function it follows that
\begin{equation*}
\begin{split}
\sum_{i = 1}^d v_{i,k} \Psi(x_i) & \ = \ \sum_{i = 1}^d v_{i,k} \sum_{j = 1}^d f_j(x_i) P_j
 \ = \ \sum_{j = 1}^d P_j \sum_{i = 1}^d v_{i,k} f_j(x_i)
 \ = \ \sum_{j = 1}^d P_j \delta_{j,k}
 \ = \ P_k\,,
\end{split}
\end{equation*}
which implies that $\Span \{\Psi(x) \mid x \in \Sigma\} \supseteq \Span \{P_k \mid k \in [d]\}.$
\end{proof}

The proof of the following proposition re-uses \cref{langequiv} from \cref{finitelabred}.
\begin{proposition}\label{nonnegreduction}
Suppose that HMM $C = (Q, \Sigma, \Psi)$ has independent functional decomposition $\Psi = \sum_{k = 1}^d f_k P_k$ and each  $P_k$ is non-negative for all $k \in [d]$. Define a set $\overline{\Sigma} = \{ a_1, \dots, a_d \}$ of fresh observations and the observation density matrix~$M$ with
$M(a_k) = P_k$ for all $k \in [d]$. Then $F = (Q, \overline{\Sigma}, M)$ is a finite-observation HMM and for any initial distributions $\pi_1, \pi_2$
\[\pi_1 \equiv_C \pi_2 \iff \pi_1 \equiv_F \pi_2. \]
\end{proposition}

\begin{proof}
It follows by \Cref{stochasticPk} that $\sum_{k = 1}^d P_k$ is stochastic. Thus $F$ defines a HMM. By \Cref{spanpreserving},  $\Span \{\Psi(x)\1^T \mid x \in \Sigma\} = \Span \{M(a)\1^T \mid a \in \overline{\Sigma}\}$ which combined with \Cref{langequiv} gives the result.
\end{proof}

\begin{example}\label{indfuncdecompex}
We use the HMM~$C$ discussed in \Cref{finiteredproblem,profileexample} to illustrate the construction of \Cref{nonnegreduction}.
The basis $\{2x\chi_{[0,1)}, 2(1-x)\chi_{[0,1)}, \frac12\chi_{[0,2)}\}$ leads to the independent functional decomposition
\begin{multline*}
\Psi = 2x\chi_{[0,1)}\begin{pmatrix}
0 & \frac12  & 0 & 0 \\
0 & 0 & 0 & 0 \\
0 & 0 & 0 & 0 \\
0 & \frac12 & 0 & 0
\end{pmatrix} + 2(1 - x)\chi_{[0,1)}\begin{pmatrix}
0 & 0  & \frac12 & 0 \\
0 & 0 & 0 & 0 \\
0 & 0 & 0 & 0 \\
0 & \frac12 & 0 & 0
\end{pmatrix} + \frac12 \chi_{[0,2)}\begin{pmatrix}
0 & 0  & 0 & 0 \\
0 & 1 & 0 & 0 \\
0 & 1 & 0 & 0 \\
0 & 0 & 0 & 0
\end{pmatrix}.
\end{multline*}
Therefore, \Cref{nonnegreduction} implies that two initial distributions $\pi_1, \pi_2 \in \RR^{|Q|}$ are equivalent in~$C$ if and only if they are equivalent in the following HMM:
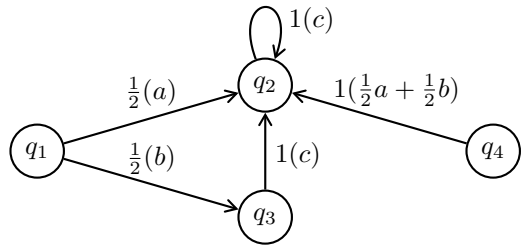
\begin{center}
\begin{tikzpicture}[scale=2.5,LMC style]
\node[state] (q1) at (0,0) {$q_1$};
\node[state] (q2) at (1.2,0.35) {$q_2$};
\node[state] (q3) at (1.2,-0.35) {$q_3$};
\node[state] (q4) at (2.4,0) {$q_4$};
\path[->] (q2) edge [loop,out=110,in=70,looseness=10] node[pos=0.75,right] {$1 (c)$} (q2);
\path[->] (q4) edge node[pos=0.4,above,yshift=1mm] {$1(\frac12 a + \frac12 b)$} (q2);
\path[->] (q3) edge node[right,pos=0.45] {$1 (c)$} (q2);
\path[->] (q1) edge node[above] {$\frac{1}{2} (a)$} (q2);
\path[->] (q1) edge node[above] {$\frac{1}{2} (b)$} (q3);
\end{tikzpicture}
\end{center}
Here, states $q_1$ and~$q_4$ are equivalent.
Hence, they are also equivalent in~$C$.
\qed
\end{example}

If an observation density matrix has an entry with pdf $2e^{-x} - 2e^{-2x}$ (which is encodable in $\pl$ due to its convex closure property), the independent functional decomposition generated by the algorithm described in the proof of \Cref{commonfuncslineardecomp} in the appendix has matrices which are not all non-negative. Therefore, \Cref{nonnegreduction} cannot be applied directly. However, given an independent functional decomposition $\Psi = \sum_{k = 1}^d f_k P_k$ and noting that $\sum_{k = 1}^d P_k$ is stochastic by \Cref{stochasticPk}, the following proposition shows that there is a small $\theta > 0$ such that $P - \theta P_k$ is non-negative for all $k \in [d]$. Furthermore, $\Span \{P_k \mid k \in [d] \} = \Span \{ P - \theta P_k \mid k \in [d] \}$. These two facts lead us to construct a finite-observation HMM using the scaled transition matrices $\frac{1}{d - \theta} (P - \theta P_k)$.
\begin{proposition}\label{reductiontime}
Let $C = (Q,\Sigma,\Psi)$ be a HMM with independent functional decomposition $\Psi = \sum_{k = 1}^d f_k P_k$. Let $P = \sum_{k = 1}^d P_k$ and
\[\theta = \min \left\{ \frac12, \frac{\min \{(P)_{i,j} \mid (P)_{i,j} > 0\}}{\max \{ \big(P_k\big)_{i,j} \mid i,j \in [Q],\ k \in [d] \} } \right\}\,. \]
Define an alphabet $\tilde{\Sigma} = \{a_1, \dots, a_d\}$ of fresh observations and the HMM $F = (Q, \tilde{\Sigma}, M)$ with $M(a_k) = \frac{1}{d - \theta}(P - \theta P_k)$. Then, for any initial distributions $\mu_1, \mu_2$
\[\mu_1 \equiv_F \mu_2 \iff \mu_1 \equiv_C \mu_2.\]
\end{proposition}
\begin{proof}
First we show that $F$ is a well-defined HMM. Matrix $\sum_{k=1}^d M(a_k)$ is stochastic as
\begin{equation}\label{sumofps}
\sum_{k=1}^d M(a_k) \ = \ \frac{1}{d - \theta}\sum_{k = 1}^d (P - \theta P_k) \ = \ \frac{dP - \theta \sum_{k=1}^d P_k}{d - \theta} \ = \ P\,,
\end{equation}
and by \Cref{stochasticPk}, $P$ is stochastic. In addition we must show that $M(a_k)$ is non-negative for each $k \in [d]$.
Since $\theta \le \frac12$, it is enough to show that $P - \theta P_k$ is non-negative for each $k \in [d]$. Suppose that $(P)_{i,j} = 0$.
Then, $\int_\Sigma \Psi_{i,j}\, d\lambda = (P)_{i,j} = 0$, which implies that $\Psi_{i,j} = 0$ since $\Psi$ is piecewise continuous.
Thus, $\sum_{k = 1}^d f_k (P_k)_{i,j} = \Psi_{i,j} = 0$.
Since $\{f_k\}_{k=1}^d$ is linearly independent, it follows that $(P_k)_{i,j} = 0$ for all $k \in [d]$ and so $(P - \theta P_k)_{i,j} = 0$.
Now suppose that $(P)_{i,j} > 0$. By the definition of $\theta$, it follows that $(\theta P_k)_{i,j} \leq (P)_{i,j}$. 
Thus, $F$ is a well defined HMM.

Observe that $\Span \left\{ P - \theta P_k \mid k \in [d] \right\} \subseteq \Span \{P_k \mid k \in [d]\}$.
The opposite inclusion follows from the fact that, by \Cref{sumofps}, we have $P \in \Span \left\{ P - \theta P_k \mid k = 1, \dots, d \right\}$.
Thus, by \Cref{spanpreserving},
\[\Span\{M(a) \mid a \in \tilde{\Sigma}\} = \Span\{P - \theta P_k \mid k \in [d]\} = \Span\{P_k  \mid k \in [d]\} = \Span\{\Psi(x) \mid x \in \Sigma\}\,.\]
Hence, the proposition follows from \Cref{langequiv}.
\end{proof}

Now we can prove \Cref{computationequivalencethm}:

\begin{proof}[Proof of \Cref{computationequivalencethm}]
Suppose the HMM $C=(Q, \Sigma, \Psi)$ is over the linearly decomposable profile language~$\Gamma$.
Let $\Gamma_0 = \{\gamma_1, \ldots, \gamma_n\}$ be the set of profiles appearing in the description of~$\Psi$.
From the description of~$\Psi$ as a matrix of coefficient-profile pairs, we can easily compute matrices $P'_1, \dots, P'_n \in \QQ^{|Q| \times |Q|}$ such that $\Psi = \sum_{i=1}^n [\![\gamma_i]\!] P'_i$.
Since $\Gamma$ is linearly decomposable, one can compute in polynomial time a subset $\{\beta_1, \dots, \beta_d\} \subseteq \Gamma_0$ such that $[\![\{\beta_1, \dots, \beta_d\}]\!]$ is linearly independent and also a set of coefficients $b_{i,k}$ such that $[\![\gamma_i]\!] = \sum_{k = 1}^d b_{i,k} [\![\beta_k]\!]$ for all $i \in [n]$.
Hence:
\[
\Psi
\ = \ \sum_{i=1}^n [\![\gamma_i]\!] P'_i
\ = \ \sum_{i=1}^n \sum_{k=1}^d [\![\beta_k]\!] b_{i,k} P'_i
\ = \ \sum_{k=1}^d [\![\beta_k]\!] \sum_{i=1}^n b_{i,k} P'_i
\]
Setting $P_k = \sum_{i=1}^n b_{i,k} P'_i$ for all $k \in [d]$, we thus obtain the independent functional decomposition $\Psi = \sum_{k=1}^d [\![\beta_k]\!] P_k$.
Now it is straightforward to compute the finite-observation HMM~$F$ from \Cref{reductiontime} in polynomial time, thus reducing the equivalence problem in~$C$ to the equivalence problem in the finite-observation HMM~$F$.
By \Cref{prop-finite-HMM} the theorem follows.
\end{proof}

\begin{example}
We illustrate aspects of the proof of \Cref{computationequivalencethm} using the HMM:
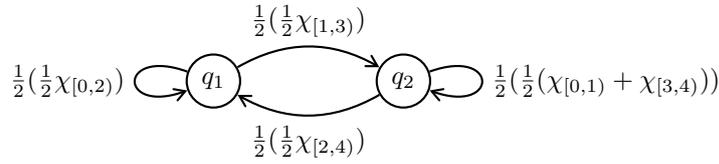
\begin{center}
\begin{tikzpicture}[scale=2.5,LMC style]
\node[state] (q1) at (0,0) {$q_1$};
\node[state] (q2) at (+1,0) {$q_2$};
\path[->] (q1) edge [loop,out=160,in=200,looseness=10] node[left] {$\frac{1}{2} (\frac{1}{2}\chi_{[0,2)})$}(q1);
\path[->] (q1) edge [bend left] node[above] {$\frac{1}{2} (\frac{1}{2}\chi_{[1,3)})$} (q2);
\path[->] (q2) edge [bend left] node[below] {$\frac{1}{2} (\frac{1}{2}\chi_{[2,4)})$} (q1);
\path[->] (q2) edge [loop,out=20,in=340,looseness=10] node[right] {$\frac{1}{2} (\frac{1}{2}(\chi_{[0,1)} + \chi_{[3,4)}))$}(q2);
\end{tikzpicture}
\end{center}
Noting that $\frac{1}{2}(\chi_{[0,1)} + \chi_{[3,4)}) =  \frac{1}{2}\chi_{[0,2)} - \frac{1}{2}\chi_{[1,3)} + \frac{1}{2}\chi_{[2,4)}$ and the set $\{\frac{1}{2}\chi_{[0,2)}, \frac{1}{2}\chi_{[1,3)},\frac{1}{2}\chi_{[2,4)}\}$ is linearly independent we obtain the independent functional decomposition
\[
 \Psi \ = \
 \frac{1}{2}\chi_{[0,2)} \begin{pmatrix}\frac12&0\\0&\frac12\end{pmatrix} +
 \frac{1}{2}\chi_{[1,3)} \begin{pmatrix}0&\frac12\\0&-\frac12\end{pmatrix} +
 \frac{1}{2}\chi_{[2,4)} \begin{pmatrix}0&0\\\frac12&\frac12\end{pmatrix}\,.
\]
According to \Cref{reductiontime}, $P = \begin{psmallmatrix}\frac12&\frac12\\\frac12&\frac12\end{psmallmatrix}$.
Further we compute $\theta = \frac12$ and $d - \theta = \frac52$ and
\begin{alignat*}{3}
& M(a) &&= \frac25 \Big[\begin{pmatrix}
\frac12&\frac12\\\frac12&\frac12
\end{pmatrix} - \frac12\begin{pmatrix}
\frac12&0\\0&\frac12
\end{pmatrix}\Big] && = \begin{pmatrix}
\frac1{10}&\frac15\\\frac15&\frac1{10}
\end{pmatrix} \\
& M(b) && = \frac25\Big[\begin{pmatrix}
\frac12&\frac12\\\frac12&\frac12
\end{pmatrix} - \frac12\begin{pmatrix}
0&\frac12\\0&-\frac12
\end{pmatrix}\Big] && = \begin{pmatrix}
\frac15&\frac1{10}\\\frac15&\frac3{10}
\end{pmatrix} \\
& M(c) && = \frac25\Big[\begin{pmatrix}
\frac12&\frac12\\\frac12&\frac12
\end{pmatrix} - \frac12\begin{pmatrix}
0&0\\\frac12&\frac12
\end{pmatrix}\Big] && = \begin{pmatrix}
\frac15&\frac15\\\frac1{10}&\frac1{10}
\end{pmatrix}.
\end{alignat*}
It follows that any initial distributions $\pi_1$ and $\pi_2$ are equivalent in $(Q, \Sigma, \Psi)$ if and only if they are equivalent in the following HMM:
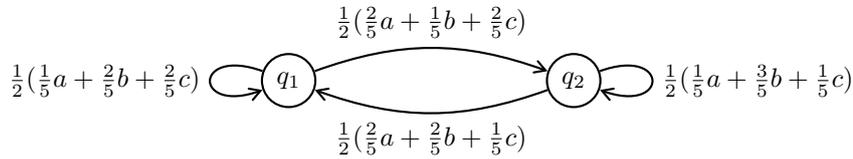
\begin{center}
	\begin{tikzpicture}[scale=2.5,LMC style]
	\node[state] (q1) at (0,0) {$q_1$};
	\node[state] (q2) at (+1.5,0) {$q_2$};
	\path[->] (q1) edge [loop,out=160,in=200,looseness=10] node[left] {$\frac12 (\frac15 a + \frac25 b + \frac25 c)$} (q1);
	\path[->] (q1) edge [bend left=20] node[above] {$\frac12 (\frac25 a + \frac15 b + \frac25 c)$} (q2);
	\path[->] (q2) edge [bend left=20] node[below] {$\frac12 (\frac25 a + \frac25 b + \frac15 c)$} (q1);
	\path[->] (q2) edge [loop,out=20,in=340,looseness=10] node[right] {$\frac12 (\frac15 a + \frac35 b + \frac15 c)$} (q2);
	\end{tikzpicture}
\end{center}
For any initial distributions $\pi_1, \pi_2 \in \QQ^2$ this can be checked with \Cref{prop-finite-HMM}.
(In this example $\pi_1 \equiv \pi_2$ holds only if $\pi_1 = \pi_2$.)
\qed
\end{example}

\begin{example} \label{ex-timing}
We also discuss an example, inspired from~\cite{braun2015robust}, where HMM non-equivalence means susceptibility to timing attacks, and HMM equivalence means immunity to such attacks.
Consider a system that emits two kinds of observations, both visible to an attacker: a function to be executed (we arbitrarily assume a choice between two functions $a$ and~$b$, and impute a probability distribution between them) and the time it takes to execute that function.
An attacker therefore sees a sequence $\ell_1 t_1 \ell_2 t_2 \ldots$, where $\ell_i \in \{a,b\}$ and $t_i \in [0,\infty)$.
In~\cite{braun2015robust} the times $t_1, t_2, \ldots$ are all identical and depend only on the secret key held by the system, but we assume in the following that the $t_i$ are drawn from a probability distribution that depends on the function ($a$ or~$b$) and the key.
We assume that with key~$i$ the execution times have uniform distributions $U[m_i^a-\frac12, m_i^a+\frac12)$ and $U[m_i^b-\frac12, m_i^b+\frac12)$.
The situation can then be modelled with the HMM below.%
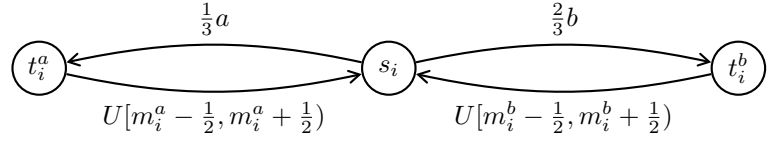
\footnote{In this case the observation set $\Sigma = [0,\infty) \cup \{a,b\}$ is a disjoint union of topological spaces and there is a natural measure space induced from the Lebesgue measure space on $[0,\infty)$ and a discrete measure on $\{a,b\}$.}
\begin{center}
\begin{tikzpicture}[scale=2.3,LMC style,yscale=0.6]
\node[state] (q1) at (2,0) {$s_i$};
\node[state] (q1a) at (0,0) {$t_i^a$};
\node[state] (q1b) at (4,0) {$t_i^b$};
\path[->] (q1a) edge [bend right=20] node[below] {$U[m_i^a-\frac12, m_i^a+\frac12)$} (q1);
\path[->] (q1) edge [bend right=20] node[above] {$\frac13 a$} (q1a);
\path[->] (q1b) edge [bend left=20] node[below] {$U[m_i^b-\frac12, m_i^b+\frac12)$} (q1);
\path[->] (q1) edge [bend left=20] node[above] {$\frac23 b$} (q1b);
\end{tikzpicture}
\end{center}
A \emph{timing leak} occurs if the attacker can glean the key from the execution times.
For example, the attacker can distinguish between keys $k_1$ and~$k_2$ if and only if states $s_1$ and $s_2$ are not equivalent.
One can check, using the algorithm we have developed in this section, that $s_1$ and $s_2$ are equivalent if and only if $m_1^a = m_2^a$ and $m_1^b = m_2^b$.
Moreover, it follows from \cref{sec-linearly-decomposable} that if instead of $U[m_1^a-\frac12, m_1^a+\frac12)$ and $U[m_2^a-\frac12, m_2^a+\frac12)$ we had two distributions with density functions from $[\![\pl]\!]$ with the same mean and the same variance, states $s_1, s_2$ would still be non-equivalent whenever the two distributions are not identical.

One may try to guard against this timing leak by ``\emph{padding}'' the execution time, so that the sum of the execution time and an added time is constant (and independent of the key).
After the execution of the function, an idling loop would be executed until the worst-case (among all keys) execution time of the functions has been reached or exceeded. Let us call this worst-case execution time $w \in (0,\infty)$.
This idling loop would take time $u>0$ in each iteration, so the total idling time is always an integer multiple of~$u$.
It is argued in~\cite{braun2015robust} that this guard is in general ineffective in that the attacker can still glean the execution time modulo~$u$.
Therefore, it is suggested in~\cite{braun2015robust} to add, in addition, a time that is uniformly distributed on $[0,u)$.

This remedy also works in our case with random execution times.
Indeed, one can show that for any independent random variables~$X,Y$, where $Y$ is distributed with~$U[0,u]$, we have that $(X + Y) \bmod u$ is distributed with~$U[0,u)$. Therefore, by adding an independent $U[0,u)$ random time to the padding described above, the times observable by the attacker now have a $U[w+u,w+2u)$ distribution, independent of the key.
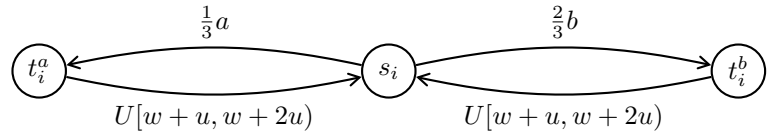
\begin{center}
\begin{tikzpicture}[scale=2.3,LMC style,yscale=0.6]
\node[state] (q1) at (2,0) {$s_i$};
\node[state] (q1a) at (0,0) {$t_i^a$};
\node[state] (q1b) at (4,0) {$t_i^b$};
\path[->] (q1a) edge [bend right=20] node[below] { $U[w+u,w+2u)$} (q1);
\path[->] (q1) edge [bend right=20] node[above] {$\frac13 a$} (q1a);
\path[->] (q1b) edge [bend left=20] node[below] {$U[w+u,w+2u)$} (q1);
\path[->] (q1) edge [bend left=20] node[above] {$\frac23 b$} (q1b);
\end{tikzpicture}
\end{center}
All states~$s_i$ are now equivalent, so the key does not leak.
\qed
\end{example}

\section{Conclusions}

We have shown that equivalence of continuous-observation HMMs is decidable in polynomial time, by reduction to the finite-observation case.
The crucial insight is that, rather than integrating the density functions, one needs to consider them as elements of a vector space and computationally establish linear (in)dependence of functions.
Therefore, our polynomial-time reduction performs symbolic computations on continuous density functions.
As a suitable framework for these computations we have introduced the notion of linearly decomposable profile languages, and we have established $\pl$ as such a profile language.

In future work, it would be desirable to extend~$\pl$ and/or develop other linear decomposable profile languages, including over sets $\Sigma$ of observations that are not real numbers.
The authors believe that the developed computational framework may be the foundation for further algorithms on continuous-observation HMMs.
For example, one may want to compute the total-variation distance of two continuous-observation HMMs.
Can Markov chains with continuous emissions be model-checked efficiently?

\bibliography{continuousalphabethmcs}

\appendix

\section{Missing Proofs}\label{alternants}

\subsection{Proof of Proposition~\ref{equivifperpspan} and Lemma~\ref{langequiv}}

\equivifperpspan*

Unlike in finite probability spaces this fact requires additional assumptions about the space. Therefore we first prove a Lemma that encapsulates these assumptions.

\begin{lemma}\label{pwcontinuousspanlemma}
Let $(Q, \Sigma, \Psi)$ be a HMM and let $\pi_1, \pi_2$ be initial distributions. As discussed in the preliminaries, $(\Sigma, \GG, \lambda)$ is a measure space such that any open set $E \in \GG$ has non-null measure. Fix $n \in \NN$. Then, $(\pi_1 - \pi_2) \Psi(w) \1^T = 0$ for all $w \in \Sigma^n$ if and only if $(\pi_1 - \pi_2)\int_{E} \Psi\,  d\lambda^n \1^T = 0$ for all $E \in \GG^n$.
\end{lemma}

\begin{proof}
The forward implication is clear.
For the converse, suppose $v \in \Sigma^n$ is such that $(\pi_1 - \pi_2) \Psi(v) \1^T > 0$.
Since $\Psi$ is piecewise continuous when restricted to $\Sigma^n$, it has a set of continuity $C^n$ as described in the preliminaries and there is a sequence of words $v_k \in C^n$ such that $\lim_{k \rightarrow \infty} \Psi(v_k) = \Psi(v)$.
As $C^n$ is an open set, there is a sequence of open balls $B(v_k, \epsilon_k) \subseteq C^n$ with $\lim_{k \to \infty} \epsilon_k = 0$.
Hence there is $k \in \NN$ such that $(\pi_1 - \pi_2) \Psi(w) \1^T > 0$ for all $w \in B(v_k, \epsilon_k)$.
By the property of $(\Sigma, \GG, \lambda)$ stated in the proposition, we have $\lambda(B(v_k, \epsilon_k)) > 0$
 and therefore $(\pi_1 - \pi_2) \int_{B(v_k, \epsilon_k)} \Psi\, d\lambda^n \1^T > 0$. A symmetrical argument can be applied in the case $(\pi_1 - \pi_2) \Psi(v) \1^T < 0$.
\end{proof}

We may now prove \Cref{equivifperpspan}.

\begin{proof}[Proof of \Cref{equivifperpspan}]
We have:
\begin{align*}
\pi_1 \equiv \pi_2 & \iff \PP_{\pi_1}(E) = \PP_{\pi_2}(E) & & \forall E \in \GG^* \\
& \iff (\pi_1 - \pi_2) \int_{E} \Psi\, d\lambda^n \1^T = 0 & & \forall E \in \GG^n,\ n \in \NN \\
& \iff (\pi_1 - \pi_2) \Psi(w) \1^T = 0 & & \forall w \in \Sigma^n,\ n \in \NN\\
& \iff (\pi_1 - \pi_2) \perp \Span \{\Psi(w) \1^T \mid w \in \Sigma^*\}\,,
\end{align*}
where the third equivalence follows from \Cref{pwcontinuousspanlemma} and is a result of $\Psi$ being piecewise continuous.
\end{proof}

Now we prove \Cref{langequiv}.

\langequiv*

\begin{proof}
Let $w = x_1\cdots x_N \in \Sigma_1^*$.
Then $\Psi_1(x_n) = \sum_{i = 1}^{I_n} \lambda_{i,n} \Psi_2(y_{i,n})$ for $n \in [N]$ and
\begin{align*}
\Psi_1(w) & \ = \ \Big(\sum_{i_1 = 1}^{I_1} \lambda_{i_1,1} \Psi_2(y_{i_1,1}) \Big) \dots \Big(\sum_{i_N = 1}^{I_N} \lambda_{i_N,N} \Psi_2(y_{i_N,N}) \Big) \\
& \ = \ \sum_{i_1 = 1}^{I_1} \dots \sum_{i_N = 1}^{I_N} \lambda_{i_1, 1} \dots \lambda_{i_N,N} \Psi_2(y_{i_1, 1}) \dots \Psi_2(y_{i_N, N}) \ \in \ \Span \{ \Psi_2(w) \mid w \in \Sigma_2^*\}.
\end{align*}
Thus, $\Span \{\Psi_1(w) \mid w \in \Sigma_1^*\} \subseteq \Span \{\Psi_2(w) \mid w \in \Sigma_2^*\}$. Therefore, by \Cref{equivifperpspan},
\begin{align*}
\pi_1 \equiv_{C_2} \pi_2 & \iff \pi_1 - \pi_2 \perp \Span \{\Psi_{2}(w) \1^T \mid w \in \Sigma_2^*\} \\
& \implies \pi_1 - \pi_2 \perp \Span \{\Psi_{1}(w) \1^T \mid w \in \Sigma_1^*\} \\
& \iff \pi_1 \equiv_{C_1} \pi_2. \tag*{\qedhere}
\end{align*}
\end{proof}

\subsection{Proof of \Cref{commonfuncslineardecomp} and an illustrating example}\label{appendixl12l13}

The main argument for \Cref{commonfuncslineardecomp} comes from two Lemmas which we state and prove first.

\alternantexistence*

\begin{proof}
	Suppose that the $f_1, \dots, f_n$ are linearly dependent. Then, there exist $\lambda_1, \dots, \lambda_n \in \RR$ that are not all $0$ such that $\sum_{i = 1}^n \lambda_i f_i(x) = 0$ for all $x \in \Sigma$. The same dependence holds for the columns of any alternant matrix for the $f_i$. This proves the forward implication.
	
	Write $M_{f_1, \dots, f_n}(x_1, \dots, x_n)$ for the alternant matrix generated by the functions $f_1, \dots, f_n$ and input points $x_1, \dots, x_n$ and let $G_{f_1, \dots, f_n} : \Sigma^n \rightarrow \RR$ be given by $G_{f_1, \dots, f_n}(x_1, \dots, x_n) = \det M_{f_1, \dots, f_n}(x_1, \dots, x_n)$.

	For the converse implication, it suffices to show that $G_{f_1, \dots, f_n} = 0$ on $\Sigma^n$ implies that $\{f_1, \dots, f_n\}$ is linearly dependent. We proceed by induction on the number~$n$ of functions. Suppose $n = 1$ with single function $f$. If for all $x \in \Sigma$, $0 = G_f(x) = f(x)$ then clearly $f = 0$ and $\{0\}$ is a linearly dependent set in any vector space.
	
	Now suppose that for $n \geq 1$ and arbitrary functions $g_1, \dots, g_n : \Sigma \rightarrow \RR$, $G_{g_1, \dots , g_n} = 0$ implies that $g_1, \dots, g_n$ are linearly dependent. Let $f_1, \dots, f_{n+1} : \Sigma \rightarrow \RR$ and $x_1, \dots, x_{n + 1}$. A Laplace expansion of $G_{f_1, \dots, f_{n+1}}(x_1, \dots, x_{n+1})$ along the first row of $M_{f_1, \dots, f_{n+1}}(x_1, \dots, x_{n+1})$ gives
	\begin{equation*}
	\begin{split}
	G_{f_1, \dots, f_{n+1}}(x_1, \dots, x_{n+1}) & = f_1(x_1) G_{f_2, \dots, f_{n+1}}(x_2, \dots, x_{n+1}) \\
	& +  \cdots \\
	& + (-1)^{n} f_{n+1}(x_1) G_{f_1, \dots, f_n}(x_2, \dots, x_{n+1}).
	\end{split}
	\end{equation*}
	Suppose $G_{f_1, \dots, f_{n+1}}(x_1, \dots, x_{n+1}) = 0$ holds for all $x_1, \dots, x_n$.
	We distinguish between two cases.
	\begin{itemize}
		\item Either there exist $x_2, \dots, x_{n+1} \in \Sigma$ such that the cofactors \[G_{f_2, \dots, f_{n+1}}(x_2, \dots, x_{n+1}), G_{f_1, f_3, \dots, f_{n+1}}(x_2, \dots, x_{n+1}), \dots, G_{f_1, \dots, f_n}(x_2, \dots, x_{n+1})\] are not all~$0$.
		This establishes a linear dependence in $f_1, \dots, f_{n+1}$.
		\item Or all cofactors are $0$ for all $x_2, \dots, x_{n+1}$.
		Then, in particular, $G_{f_2, \dots, f_{n+1}}(x_2, \dots, x_{n+1}) = 0$ for all $x_2, \dots, x_{n+1}$.
		By the induction hypothesis it follows that the functions $f_2, \dots, f_{n+1}$ are linearly dependent.
		Hence, so are $f_1, \dots, f_{n+1}$.
	\end{itemize}
	In either case it follows that $f_1, \dots, f_{n+1}$ are linearly dependent.
\end{proof}

\linindepofcommonfuncs*

\begin{proof}
	Assume that there is a linear dependence\begin{equation*}
	\sum_{i = 1}^{I} r_i m_i(x) + \sum_{j = 1}^J s_j g_j(x) + \sum_{k = 1}^K t_k e_k(x) = 0 \quad \forall x \in \RR.
	\end{equation*}
	By reordering if necessary, we may assume that the exponential functions $e_1, \dots, e_K$ have strictly decreasing rates $\lambda_1 > \dots > \lambda_K$. The function $e_K$ tends to $0$ at the slowest rate out of all other functions in the linear dependence and so
	\begin{equation*}
	\lim_{x \rightarrow \infty}  \frac{1}{e_K(x)}\Big[\sum_{i = 1}^{I} r_i m_i(x) + \sum_{j = 1}^J s_j g_j(x) + \sum_{k = 1}^K t_k e_k(x)\Big] = t_K\,,
	\end{equation*}
	which implies that $t_K = 0$. Repeating this argument for decreasing $k \in [K]$ it follows that $t_1 = \dots = t_K = 0$ and therefore
	\begin{equation*}
	\sum_{i = 1}^{I} r_i m_i(x) + \sum_{j = 1}^J s_j g_j(x) = 0 \quad \forall x \in \RR.
	\end{equation*}
	Suppose the Gaussian functions $g_1, \dots, g_J$ have mean and standard deviation $\mu_1, \dots, \mu_J$ and $\sigma_1, \dots, \sigma_J$, respectively. By defining the ordering $g_i <_{\text{lex}} g_j$ if and only if $\sigma_i < \sigma_j \lor (\sigma_i = \sigma_j \land \mu_i < \mu_j)$ we may assume without loss of generality that $g_1 <_{\text{lex}} \dots <_{\text{lex}} g_J$. It follows that for $1 \leq j < J$ the ratio
	
	\begin{equation*}
	\begin{split}
	\frac{g_j(x)}{g_J(x)} & = \left. \frac{1}{\sigma_j\sqrt{2\pi}} \exp\Big[ - \frac{(x - \mu_j)^2}{2\sigma_j^2}\Big] \middle/ \frac{1}{\sigma_J\sqrt{2\pi}} \exp\Big[- \frac{(x - \mu_J)^2}{2\sigma_J^2}\Big] \right. \\
	& = \frac{\sigma_J}{\sigma_j} \exp\Big[\frac{1}{2}\Big( \Big(\frac{1}{\sigma_J^2} - \frac{1}{\sigma_j^2} \Big)x^2 - 2\Big( \frac{\mu_J}{\sigma_J^2} - \frac{\mu_j}{\sigma_j^2}\Big)x  + \Big( \frac{\mu_J^2}{\sigma_J^2} - \frac{\mu_j^2}{\sigma_j^2}\Big)\Big)\Big]\\
	& \rightarrow 0 \text{ as } x \rightarrow \infty\,,
	\end{split}
	\end{equation*}
	as $g_j <_{\text{lex}} g_J$ implies that the dominant polynomial coefficient in the exponent is always negative. Any Gaussian density function tends to $0$ slower than any interval-domain monomial at $+\infty$, so similarly to the exponential densities,
	\begin{equation*}
	0 = \lim_{x \rightarrow \infty} \frac{1}{g_J(x)}\Big[\sum_{i = 1}^{I} r_i m_i(x) + \sum_{j = 1}^J s_j g_j(x)\Big] = s_J\,.
	\end{equation*}
	By repeating this argument for decreasing $j \in [J]$, we obtain $s_1 = \dots = s_J = 0$. It remains to show the remaining interval-domain monomials are linearly independent. Since $H$ is finite, all interval-domain monomials on $[a,b)$ have a maximum exponent $R$. The intervals are disjoint so it suffices to consider a single interval $[a,b)$ and show that the set of monomials $\{x^k\chi_{[a,b)} \mid k \in \{0, \dots, R\}\}$ is linearly independent.
	Consider the alternant matrix for $1\chi_{[a,b)}, x\chi_{[a,b)}, \dots, x^R\chi_{[a,b)}$ and distinct input points $x_1, \dots, x_{R + 1} \in [a,b)$. This matrix is a Vandermonde matrix and by \cite[p.9]{milne33} has full rank. Therefore, by \Cref{alternantexistence} the set $\{1\chi_{[a,b)}, x\chi_{[a,b)}, \dots, x^R\chi_{[a,b)}\}$ is linearly independent.
\end{proof}

\commonfuncslineardecomp*

\begin{proof}
Let $\Gamma_0 \subseteq \pl$ be a finite set of profiles.
Any profile in~$\Gamma_0$ encodes a linear combination of Gaussians, exponentials and interval-domain monomials.
Collect in $G_0$ and~$E_0$ the profiles of Gaussians and exponentials, respectively, that appear in the description of at least one profile in~$\Gamma_0$.
By sorting the start and end points of the intervals (that appear in the interval-domain monomials) in~$\Gamma_0$, we compute a finite set~$H$ of disjoint intervals such that every interval appearing in~$\Gamma_0$ is a union of intervals in~$H$.
Further, collect in~$N$ the set of degrees of monomials in~$\Gamma_0$.
Then we compute a set of profiles $M_0$ such that $[\![M_0]\!] = \{x^n \chi_{[a,b)} \mid n \in N,\ [a,b) \in H\}$.
By \cref{linindepofcommonfuncs}, the set $[\![G_0 \cup E_0 \cup M_0]\!]$ is linearly independent.
We compute the (unique) coordinates of all functions in~$[\![\Gamma_0]\!]$ in terms of that basis.

With these coordinates at hand, we compute a subset $\mathcal{B} \subseteq \Gamma_0$ such that $[\![\mathcal{B}]\!]$ is a basis of $\Span [\![\Gamma_0]\!]$ as follows.
Starting with the $\mathcal{B} = \emptyset$, go through $\Gamma_0$ one by one; whenever a profile $\gamma \in \Gamma_0$ is such that $[\![\mathcal{B} \cup \{\gamma\}]\!]$ is linearly independent then add $\gamma$ to~$\mathcal{B}$.
The check for linear independence can be performed in terms of the computed coordinates of $[\![\Gamma_0]\!]$ in the basis $[\![G_0 \cup E_0 \cup M_0]\!]$.
For the final set $\mathcal{B}$ we have that $[\![\mathcal{B}]\!]$ is a basis of $\Span [\![\Gamma_0]\!]$.
The coefficients that express $[\![\Gamma_0]\!]$ as a linear combination of~$[\![\mathcal{B}]\!]$ can be computed similarly.
All computations referred to in this proof are polynomial-time.
\end{proof}

\begin{example}\label{lineardecompexample}
We illustrate the proof of \Cref{commonfuncslineardecomp} using the HMM discussed in \Cref{finiteredproblem,profileexample,indfuncdecompex}. Recall the encoding of $\Psi$ is given as the matrix

\[
\begin{pmatrix}
0 & (\frac12, \gamma_1)  & (\frac12,\gamma_2) & 0 \\
0 & (1,\gamma_3) & 0 & 0 \\
0 & (1,\gamma_3) & 0 & 0 \\
0 & (1,\gamma_4) & 0 & 0
\end{pmatrix}
\]
with $\gamma_1, \gamma_2, \gamma_3, \gamma_4 \in \pl$ and
$[\![\gamma_1]\!] = 2x\chi_{[0,1)}$ and
$[\![\gamma_2]\!] = 2(1-x)\chi_{[0,1)}$ and
$[\![\gamma_3]\!] = \frac12 \chi_{[0,2)}$ and
$[\![\gamma_4]\!] = \chi_{[0,1)}$.
Clearly $\Gamma_0 = \{\gamma_1, \gamma_2, \gamma_3, \gamma_4\} \subset \pl$. By ordering the start and end points in $[0,2), [1,2)$ we compute $H = \{[0,1), [1,2)\}$.
The set of degrees is $N = \{0,1\}$.
We then compute the set~$M_0$ of profiles such that $[\![M_0]\!] = \{\chi_{[0,1)}, x\chi_{[0,1)},\chi_{[1,2)}, x\chi_{[1,2)}\}$
and express $[\![\gamma_1]\!], \dots, [\![\gamma_4]\!]$ as vectors of coordinates with respect to the basis~$[\![M_0]\!]$:
\begin{alignat*}{3}
&[\![\gamma_1]\!] && = 2x\chi_{[0,1)}      && = (0, 2, 0, 0) \\
&[\![\gamma_2]\!] && = 2(1-x)\chi_{[0,1)}  && = (2, -2, 0, 0) \\
&[\![\gamma_3]\!] && = \frac12\chi_{[0,2)} && = (\frac12, 0, \frac12, 0) \\
&[\![\gamma_4]\!] && = \chi_{[0,1)}        && = (1,0,0,0)
\end{alignat*}
We then compute a basis for this set of vectors: $\{(0, 2, 0, 0), (2, -2, 0, 0), (\frac12,0,\frac12,0)\}$. This implies that with $\mathcal{B} = \{\gamma_1, \gamma_2, \gamma_3\}$, the set $[\![\mathcal{B}]\!]$ is a basis for $\Span [\![\Gamma_0]\!]$.
Since $(1, 0, 0, 0) = \frac12(0, 2, 0, 0) + \frac12(2, -2, 0, 0)$, we express $[\![\gamma_4]\!]$ in terms of~$[\![\mathcal{B}]\!]$ by $[\![\gamma_4]\!] = \frac12 [\![\gamma_1]\!] + \frac12 [\![\gamma_2]\!]$.
\qed
\end{example}

\end{document}